\documentclass[10pt, conference, compsocconf, letterpaper]{IEEEtran}

\usepackage[justification=centering]{caption}
\usepackage[normalem]{ulem}
\usepackage{amsmath}
\usepackage{algpseudocode}
\usepackage{amsfonts}
\usepackage{comment}
\usepackage{caption}
\usepackage{eqparbox}
\usepackage{graphicx}
\usepackage{algorithm}						
\usepackage{url}   
\usepackage[noadjust]{cite}   
\usepackage{slashbox} 	
\usepackage{epstopdf}
\usepackage{subfig}
\usepackage{float}
\usepackage{afterpage}
\usepackage{enumerate}
\usepackage{multirow}
\usepackage{amsthm}

\usepackage{mathtools}

\DeclarePairedDelimiter{\floor}{\lfloor}{\rfloor}

\newcommand{\E}{\mathbb{E}}
\newcommand{\found}{\emptyset}

\renewcommand{\binom}{\mathbb{B}}
\renewcommand{\over}{\Xi}
\newcommand{\best}{\Upsilon}

\allowdisplaybreaks

\usepackage{color}
\newcommand{\NOT}[1]{}
\newcommand{\NOTE}[1]
{
  {\footnotesize\it
    \begin{center}
      \begin{tabular}{|c|}
        \hline
        \parbox{0.85\columnwidth}{
          \medskip
          #1
          \medskip} \\
        \hline
      \end{tabular}
    \end{center}
    }
}

\newtheorem{proposition}{Proposition}
\newtheorem{theorem}{Theorem}

\begin{document}

\title{A Lightweight Approach for Improving the Lookup Performance in Kademlia-type Systems \\ \vspace{5mm} \LARGE{Technical Report}}

\author{
\IEEEauthorblockN{Hani Salah$^{\dag}$ \qquad Stefanie Roos$^{\ddag}$ \qquad Thorsten Strufe$^{\ddag}$}
\IEEEauthorblockA{$^{\dag}$TU Darmstadt, Germany \qquad $^{\ddag}$TU Dresden, Germany\\
hsalah(at)cs.tu-darmstadt.de, \texttt{firstname.lastname}(at)tu-dresden.de}
}


\maketitle
\begin{abstract}
Discovery of nodes and content in large-scale distributed systems is generally based on Kademlia, today. Understanding Kademlia-type systems to improve their performance is essential for maintaining a high service quality for an increased number of participants, particularly when those systems are adopted by latency-sensitive applications.

This paper contributes to the understanding of Kademlia by studying the impact 
of \emph{diversifying} neighbours' identifiers within each routing table bucket on the lookup performance. We propose a new, yet backward-compatible, neighbour selection scheme 
that attempts to maximize the aforementioned diversity. The scheme does not cause additional overhead except negligible computations for comparing the diversity of identifiers. We present a theoretical model for the actual impact of the new scheme on the lookup's hop count and validate it against simulations of three exemplary Kademlia-type systems. We also measure the performance gain enabled by a partial deployment for the scheme in the real KAD system. The results confirm the superiority of the systems that incorporate our scheme.\\
\end{abstract}
\IEEEpeerreviewmaketitle


\section{Introduction} \label{sec:intro}

The family of Kademlia-type \cite{Maymounkov02Kademlia} systems represents the most widely used and deployed type of Distributed Hash Tables (DHT). These systems are adopted, for the discovery of nodes and content, by several peer-to-peer (P2P) file-sharing applications, accounting for multi-million users today 
\cite{wang2013measuring, salah13capturing}, like Bittorrent 
 and eMule
. Implementations of Kademlia have been also experimented as communication overlays in video streaming applications \cite{jimenez11kademlia} and botnets \cite{starnberger08overbot, holz08measurements}.

The routing (also called \emph{lookup}), the key operation in these systems, has received a great deal of attention from the research community in the last years. Some studies (e.g. \cite{li05comparing, stutzbach06improving, steiner10evaluating, falkner07profiling, crosby07analysis, Jimenez2011subsecond, roos13comprehending}) identified limitations in the \emph{standard} system designs, thus raised doubts about their suitability, particularly, for latency-sensitive applications. 
The authors accordingly proposed modifications to the standard designs in order to improve the system performance. 
Considered performance metrics are the hop count, the latency of the lookup, and the overhead in terms of the number of sent messages.

We aim to further improve the performance based on concrete theoretical models. 
In contrast to the above studies, we do not focus on optimizing the parameters governing routing table structure and lookup mechanism. Rather, the goal is to develop a low-overhead scheme, which can be integrated easily into all existing designs. 

Towards this end, we study a previously disregarded lookup performance factor -- the diversity of neighbours' identifiers within each routing table bucket. 
Consequently, we propose, model, and evaluate a new neighbour selection scheme that attempts to maximize this diversity. The scheme is  compatible not only with the standard Kademlia and its variations, but also with previously proposed improvements. It thus can be combined to any of them to improve their achieved performance further. It does not change the standard routing protocol nor the routing table structure, and it does not cause additional communication overhead. Only slight changes in the standard routing table's maintenance processes are required, causing only a negligible extra computational overhead. 

\vspace{3pt}
Our main contributions can be summarized as follows:

\begin{itemize}
\item We propose a new neighbour selection scheme to reduce the average lookup's hop count in Kademlia-type systems, with almost no extra overhead. 
\item We develop a theoretical model (extending a prior work \cite{roos13comprehending}) to asses the impact of the proposed scheme on the lookup performance.
\item We evaluate the scheme using both the model and extensive simulations of three Kademlia-type systems. 
\item We measure the impact of our scheme on the lookup performance of modified KAD clients in the real KAD system.
\end{itemize}
 
The model predictions and simulation results, which agree very closely to each other, show that the new scheme improves the lookup performance in form of reduction in the average hop count, and thus in the number of sent messages. 
The improvement applies also for the measurement results of the KAD clients that incorporates our scheme.

The remainder of this paper is structured as follows: We give an overview of Kademlia and its variations in Sec.~\ref{sec:kademlia}, and then discuss the related work in Sec.~\ref{sec:related}. Next, Sec.~\ref{sec:solution} presents an overview of our proposed scheme, Sec.~\ref{sec:model} describes our model, and Sec.~\ref{sec:results} discusses the results. Finally, Sec.~\ref{sec:conclusion} concludes the paper.


\section{Kademlia-type Systems} \label{sec:kademlia}

Kademlia \cite{Maymounkov02Kademlia} is a structured peer-to-peer (P2P) system. It uses a $b$-bit identifier space from which the identifiers of nodes and objects are assigned. Nodes store key-value (key-object) pairs, such that the nodes at the closest distance to an identifier are responsible for storing it. The distance between two identifiers is defined as the XOR of their values.  

Each node $v$ stores the identifiers and addresses of other nodes, also called neighbours or contacts\footnote{From here onwards, these two terms are used interchangeably.}, in a $b$-level tree-structured routing table. Each level in the routing table consists of so called $k$-buckets, such that each bucket stores up to $k$ known contacts that share a common prefix with $v$'s identifier. Contacts that represent nodes which have left the system are called \emph{stale}. 

Kademlia implements a key-based routing protocol: To route a message from a node $v$ to a target identifier $x$, 
$v$ picks $\alpha$ known contacts that are closest to $x$ and sends them lookup requests in parallel. Every queried contact that is online replies with the set of $\beta$ contacts that are locally known as being closest to $x$, thus extending $v$'s set of candidate contacts. This process iterates until no further contacts closer to $x$ are discovered or a timeout is held. The original Kademlia paper suggests $k=20$ and $\alpha=3$.

In order to mitigate the effect of churn, each node performs \emph{maintenance} processes for its routing table buckets. In practice\footnote{This is how it is implemented, for example, in the eMule software: http://www.emule-project.net.}, two periodic \emph{maintenance processes} are performed: The first process aims to increase the amount of contacts that are stored in the routing table by searching for potential new contacts belonging to low populated buckets. The second process aims to keep the routing information up-to-date by checking if the stored contacts are still responsive and removing stale ones. Long-lived contacts are checked less frequently than newer ones. This preference for old contacts is based on the observation that the longer a contact has been online, the more likely it is to remain online in the future \cite{saroiu01measur, Maymounkov02Kademlia}. 

The above design is the basis for a family of Kademlia-type systems. In this paper, we focus on three of them.  
The mainline implementation of BitTorrent (MDHT) integrates one of those systems for nodes discovery. In MDHT, the routing table (Fig.\,\ref{fig:routing}(a)) includes a single k-bucket per level. uTorrent, the most popular MDHT implementation uses $k=8$, $\alpha=4$, and $\beta=1$ \cite{Jimenez2011subsecond}. 

\begin{figure*} \centering
\captionsetup{font=scriptsize}
      \includegraphics[width=0.95\linewidth, height=0.20\textheight]{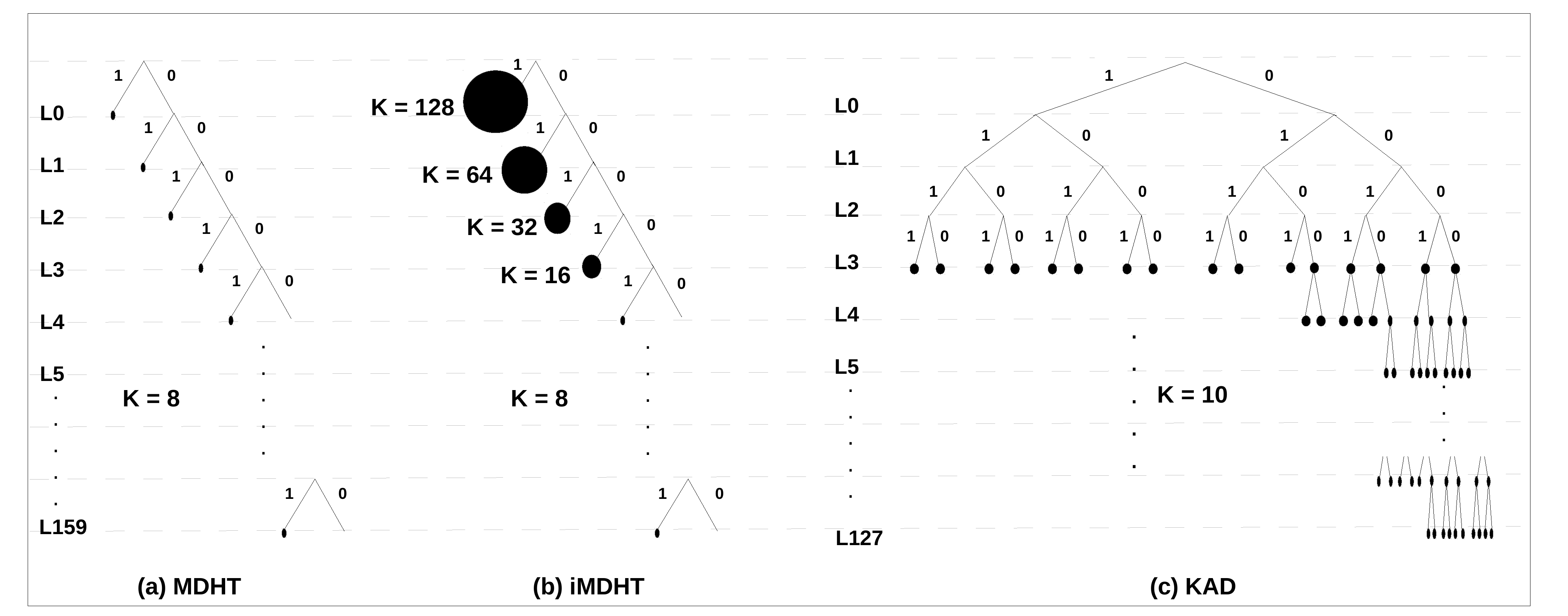}
   \caption{The routing table structures of three Kademlia-type systems (adapted from: \cite{wang08attacking}): (a) MDHT, (b) iMDHT, and (c) KAD.}
   \label{fig:routing}
\end{figure*} 

Considering the fractions of the identifier space that are covered at each routing table level $i$, Jimenez et al. \cite{Jimenez2011subsecond} introduced a variation of MDHT implementing variable bucket sizes (iMDHT) to increase the distance reduction at each hop. The bucket size is chosen to be 128, 64, 32, and 16 for the buckets at levels $i \in (0..3)$ respectively, and 8 for the rest (Fig.\,\ref{fig:routing}(b)). Both MDHT and iMDHT use $b=160$.

KAD, the DHT used by the popular file-sharing application eMule, uses  $b=128$, $k=10$, $\alpha=3$, and $\beta=2$. It implements a different routing table structure: As shown in Fig.\,\ref{fig:routing}(c), starting from the fourth level, the routing table includes multiple buckets per level, grouping contacts according to the first $l \in \{3,4\}$ bits after the first varying bit. Consequently, the difference between the common prefix length of the current hop and the next hop to $x$, called the \emph{bit gain}, is at least $l$.


\section{Related Work} \label{sec:related}

Motivated by the high popularity of Kademlia and its variations, those systems have been the subject of many studies in the past few years. In this Section, we discuss only those studies which focused on the lookup process, or those which proposed improvements for the standard system design.

Crosby and Wallach \cite{crosby07analysis} measured the lookup latency in MDHT and Azureus (the DHT that is used by the Vuze Bittorent client). They reported high latency values, and attributed this to the observed high ratio of stale contacts in the routing tables. 
Similarly, Stutzbach and Rejaie \cite{stutzbach06improving} analysed the lookup process and measured the lookup latency in KAD. 

Several studies investigated the possibility to improve the lookup performance. The approach by Falkner et al. \cite{falkner07profiling} adapted the lookup parameters at runtime according to the number of expected lookup response messages. Their design reduced the median lookup latency but at the same time increased the lookup overhead. Steiner et al. \cite{steiner10evaluating}, in addition to analysing the lookup latency in KAD by evaluating the impact of both external factors (e.g. lookup message RTT) as well as the lookup parameters, achieved an improved lookup latency by coupling the lookup with the content retrieval process.  More recently, Jimenez et al. \cite{Jimenez2011subsecond} suggested several modifications to MDHT, and achieved better lookup latency with low additional overhead. 

A number of other studies succeed to reduce the lookup costs (measured by number of lookup messages or latency) via: caching \cite{einziger13kaleidoscope, wang06dicas}, geographical proximity \cite{kaune08embracing, castro02exploiting}, or recursive lookup \cite{heep10r}. 

In this paper, we propose to improve the lookup performance, in form of reduction in the average hop count, by adapting the standard neighbour selection scheme. Although the approach differs from earlier improvements, it is orthogonal and hence compatible with them. 
 
The model that we present in this paper extends on a prior model of Kademlia-type systems \cite{roos13comprehending}, which allows a very accurate prediction of the routing overhead.
Though theoretical analysis of P2P routing performance is widely studied, traditionally only asymptotic bounds have been derived (e.g. \cite{stoica2001chord,rowstron2001pastry, Maymounkov02Kademlia,malkhi2002viceroy}).
The few studies deriving exact formulas commonly only consider the average routing length, special cases such as bijective mapping from identifiers to nodes, and are of limited accuracy when compared to measurements or simulations (e.g.
\cite{stutzbach06improving,spognardi2006formal, rai2007performance}).
In particular, \cite{spognardi2006formal, rai2007performance} model P2P routing using a Markov chain approach similar to 
the one suggested in \cite{roos13comprehending}, but are restricted to systems without parallelism.


\section{Improving the Lookup Performance} \label{sec:solution}

In this Section, we introduce an approach for improving the lookup performance in Kademlia-type systems: We give an overview of the idea in Sec.~\ref{subsec:idea} and then validate the main assumption on which it is based (against results obtained from a real Kademlia-type system) in Sec.~\ref{subsec:real_diverse}. In Sec.~\ref{subsec:increasing}, we describe how the approach can be implemented in practice. 

\subsection{Idea} \label{subsec:idea}

Our approach is based on a new neighbour selection scheme. As described in Sec.~\ref{sec:kademlia}, the $k$ contacts (i.e. neighbours) that each routing table bucket can store share a common prefix. This means that they all belong to a specific range of identifiers, thus represent a certain region of the identifier space. 
This way, without further restrictions, multiple contacts in the bucket can have very close identifiers (e.g. represent contiguous positions in the respective identifier space region), whereas there are other portions of the region not closely covered by the stored contacts.

By storing a large number of contacts from only one portion of the region, the node's view for the respective identifier space region is  narrowed.
Consequently, we propose to improve the lookup performance by \emph{widening} this view. That is, the node should try to maximize the \emph{diversify} of the identifiers that are stored in each of its routing table buckets independently, 
by choosing contacts such that their identifier prefixes are maximally diverse. 
Then the expected common prefix length of the closest contact to an arbitrary target identifier is maximized, which should lead to a lower average number of hops.

We now show that the expected bit gain, i.e. the difference between the node and
the closest contact in its routing table to a target identifier $x$, is increased by maximizing the diversity. A general model of the actual impact on the average hop count is presented in Section \ref{sec:model}.
 \begin{theorem}
 \label{thm:bitgain}
 Consider a $k$-bucket such that contacts in the bucket offer a bit gain of at least $l$. 
 The expected
bit gain $bg^{div}$ offered by the closest contact to $x$ in that $k$-bucket when maximizing
the diversity  is at least as big as $bg^{norm}$, the expected bit gain for the standard contact selection.
\end{theorem}  
\begin{proof}
The cumulative distribution function (CDF) of the expected bit gain of one contact chosen uniformly at random
from all identifiers in the $k$-bucket is given by 
\begin{align*}
F^l(i)=\begin{cases}
0, & i < l \\
1-1/2^{i-l}, & i \geq l
\end{cases}
\end{align*}
because there is a guaranteed improvement of $l$
and the probability for every further bit to agree with the respective bit of $x$ is \begin{footnotesize}$\dfrac{1}{2}$\end{footnotesize}.
Note that the CDFs for the maximum of independent random variables $X_1, \ldots , X_m$
with CDFs $F_1, \ldots , F_m$ is $F(x)= 1 - \prod_{i=1}^m F_i(x)$.
Furthermore, the expected value of a random variable $X$ with values in $\mathbb{N}_0$ and CDF $F$ is
$\E(X)=\sum_{i=0}^\infty P(X > i)=\sum_{i=0}^\infty 1-F(i)$.
The distribution of the maximum bit gain of $k$ contacts when selecting contacts uniformly at random from all nodes
suitable for a bucket is thus
\begin{align*}
bg^{norm} = l+\sum_{i=l+1}^\infty 1 - F^l(i)^k.
\end{align*}
When maximizing the diversity of the $q=\lfloor \log k \rfloor$ additional bits,
there is one contact for each $2^q$ bit sequences as well as $k-2^q$ contacts 
chosen uniformly at random.
Thus, the guaranteed bit gain is $q$. The closest contact in the routing
table is either the one contact guaranteed to agree with $x$
in those $q$ bits or one of the contacts chosen uniformly at random.
The CDF of the first is given by $F^l_A(x)=F^l(x-q)$ and the CDF of the
maximum bit gain of the contacts chosen uniformly at random is $F^l_B(x)=\left(F^l(x)\right)^{k-2^q}$.
Thus, the CDF of the total bit gain is $F^l_A(x)F^l_B(x)$,
so that the expected bit gain is given by 
\begin{align*}
bg^{div} = l+q + \sum_{i=l+q+1}^\infty 1 - F^l(i-q)\left(F^l(i)\right)^{k-2^q}.
\end{align*}
$bg^{div}$ presents an upper bound on $bg^{norm}$ because 
\begin{align*}
bg^{norm} &=  l+\sum_{i=l+1}^\infty 1 - \left(F^l(i)\right)^{2^q}\left(F^l(i)\right)^{k-2^q} \\
&\leq l+q + \sum_{i=l+q+1}^\infty 1 - \left(F^l(i)\right)^{2^q}\left(F^l(i)\right)^{k-2^q} \\
&\leq l+q + \sum_{i=l+q+1}^\infty 1 - \left(1-\frac{2^q}{2^{i-l}}\right)\left(F^l(i)\right)^{k-2^q}  \\
&= l+q + \sum_{i=q+1}^\infty 1 - \left(1-\frac{1}{2^{i-l-q}}\right)\left(F^l(i)\right)^{k-2^q} \\&= bg^{div}.
\end{align*} \end{proof}

Note that Theorem \ref{thm:bitgain} only shows that the expected bit gain of the closest
contacts in one node's routing table is increased. Lookup parallelism ($\alpha>1$) and further contacts in the lookup response ($\beta > 1$) are not considered so far. Nevertheless, the above result motivates an in-depth analysis of the contact selection scheme, which we present in the following sections.

From here onwards, we use the term \emph{diversity degree} for a bucket to indicate how diverse are the prefixes of the identifiers stored in it. It is measured by the number of distinct $\floor*{log~k}$ bits after the bucket's common prefix, resulting in a maximal degree of  $2^{\floor*{log~k}}$. That is 3 bits in MDHT and KAD, whereas iMDHT considers 7, 6, 5, and 4 bits for buckets at levels $i \in (0..3)$ respectively, and 3 bits for the rest. Fig.\,\ref{fig:uniformity} shows an exemplary MDHT routing table: The diversity degrees of buckets \textit{A} and \textit{B} are $8$ (i.e. the maximal value in MDHT) and $4$, respectively.

\begin{figure} \centering
\captionsetup{font=scriptsize}
      \includegraphics[width=0.95\linewidth, height=0.18\textheight]{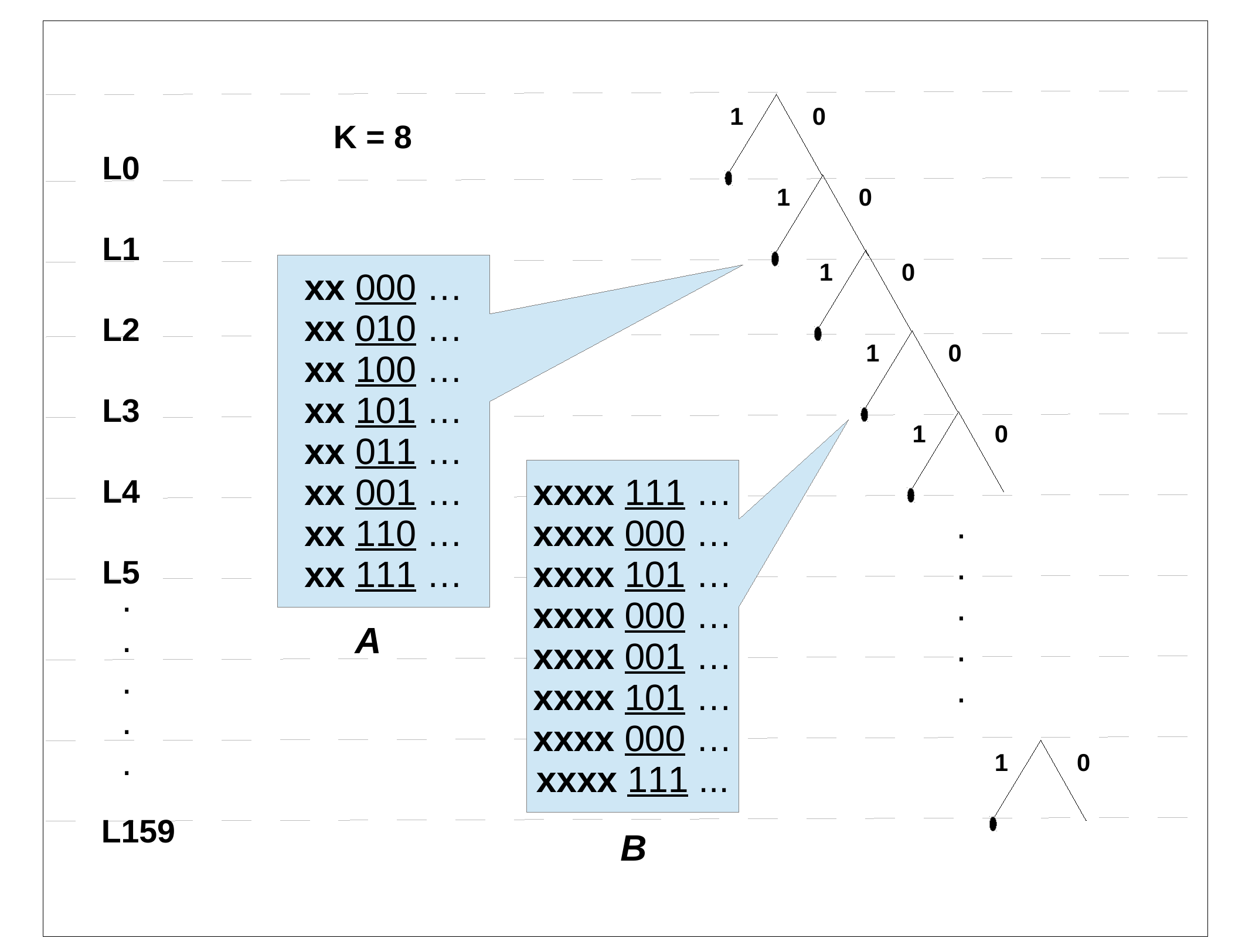}
   \caption{An exemplary MDHT routing table: Buckets \textit{A} and \textit{B} are located at the second and fourth levels (i.e. levels 1 and 3), thus have common prefixes of lengths 2 and 4, respectively. Considering the first three bits after the common prefix, their respective diversity degrees are $8$ (i.e. the maximal value in MDHT) and $4$.}
   \label{fig:uniformity}
\end{figure}

Note that the idea above is similar to the KAD's routing table structure, which divides contacts having the same common prefix length with the routing table owner into buckets according to the first 4 bits after the common prefix. However, our approach is more flexible, since it does not restrict the number of prefixes per bucket, but rather selects more diverse prefixes if possible, allowing for a less diverse contact selection if maximal diversity is not achievable. 


\subsection{Diversity Degrees in Real Kademlia-type Systems} \label{subsec:real_diverse}

We here aim to validate our aforementioned assumption that the bucket in standard Kademlia-type systems is likely to store multiple contacts having very close identifiers.   
Towards this end, we downloaded routing tables of randomly selected online KAD nodes using an accurate KAD crawler \cite{salah13capturing}, and we then analysed the diversity degrees of the contact identifiers contained in their buckets. We restricted our analysis only to the buckets located at the fourth routing table level, for two reasons: First, given the routing table structure of KAD (Fig.\,\ref{fig:routing}(c)), those buckets jointly are used, on average, in \begin{footnotesize}$\dfrac{11}{16}$\end{footnotesize} of the lookup requests, hence represent the most important part of the routing table.
Second, there exist with very high probability nodes in the system that can fill those buckets with all possible prefixes, enabling those buckets to achieve higher completeness (as shown in \cite{stutzbach06improving}) and higher diversity degrees than buckets at lower levels. 

The results that we discuss here represent 1,505,658 buckets. We classify them, by the number of contacts they contain, into three groups: (\emph{i}) 170,262 buckets containing eight contacts each, (\emph{ii}) 271,585 buckets containing nine contacts each, and (\emph{iii}) 1,063,811 buckets containing ten contacts each (i.e. complete buckets). For each group, we computed the CDF of buckets with a diversity degree $\leq m$. Recall that in KAD $m \in (1..8)$. 

Fig.\,\ref{fig:measured_unif_cdf} shows the results: 42\% to 67\% of the buckets have $m \leq 4$ (i.e. half the maximal degree), and at most 8\% have the maximal degree. These results confirm our aforementioned assumption about diversity degrees in the real systems.

\begin{figure} \centering
\captionsetup{font=scriptsize}
   \includegraphics[width=0.95\linewidth, height=0.18\textheight]{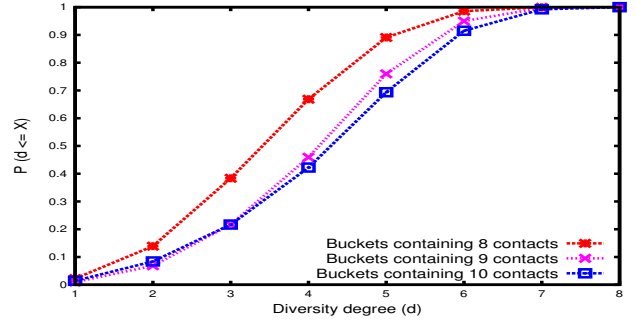}
   \caption{CDF of the diversity degrees of measured standard KAD buckets.}
   \label{fig:measured_unif_cdf}
\end{figure}

\subsection{Implementation} \label{subsec:increasing}

The aforementioned idea to improve the lookup performance can be implemented in real Kademlia-type systems by only slightly modifying the standard routing table's maintenance processes that we discussed in Sec.~\ref{sec:kademlia} (i.e. modifying the neighbour selection), without changing their frequency. More precisely, when a node decides to find new contacts to insert them to an incomplete bucket, or to replace stale contacts, it selects contacts whose identifiers increase the  diversity degree of that bucket. This way, the approach does not require changing the original routing table structure nor the routing protocol, and it does not induce any additional overhead, except computing the diversity degrees over the identifiers. 


\section{Model} \label{sec:model}

In this Section, we analytically derive the hop count distribution in Kademlia-type systems
both for the standard contact selection scheme as well as for our modified scheme
. More precisely, we determine the probability that the path needed to find the closest node to a target identifier $x$ is of length $h$ for both schemes. 

The model that we present here extends our prior work \cite{roos13comprehending} such that:  (\emph{i}) it allows queries for arbitrary identifiers rather than for only node identifiers, and (\emph{ii}) it integrates our contact selection scheme, which improves the diversity of identifiers, into the model. In Sec.~\ref{sec:modelreview}, we review the main components of the original model \cite{roos13comprehending}. In Sec.~\ref{sec:modelsuccess} and Sec.~\ref{sec:modeltransition}, we describe in details our modifications on the original model to achieve (\emph{i}) and (\emph{ii}), respectively.

\subsection{Model Principles} \label{sec:modelreview}
We analytically answer the question:
Given an arbitrary target identifier $x$ and a Kademlia-type system, how likely is it
to discover the responsible node within $h$ hops for all $h$?

To this end, we model the routing process as a Markov chain, such that states of the chain represent the 
common prefix length of the $\alpha$ closest nodes with the target identifier $x$. 
Due to the prefix-based routing table structure, the common prefix length is sufficient
to determine the transition probabilities, i.e. the probability to change
from one set of common prefix lengths to either the terminal state of discovering the target or 
another set of  common prefix lengths.
The probability to reach the target in $h$ hops is then given by the probability that
the Markov chain is in the terminal state after $h$ steps.



%

Recall that the assignment of a contact to a bucket is made on the basis of the common prefix length with the routing table owner. The corresponding distance function, which we refer to as the \emph{bit distance}, is%
\begin{align}%
dist(x,y) = b - commonprefixlength(x,y).%
\end{align}%

Following \cite{roos13comprehending}, we formally characterize a Kademlia-type system by the identifier space size $b$,  and the routing table parameters $k$ (the bucket size) and $L$ (the number of buckets per level).
The $d$-th entry $k_d$ of the $b$-dimensional vector $k$ gives the bucket size for contacts at bit distance $d$ to the routing table owner. For example, in iMDHT we have
\begin{align*}
k_d = \begin{cases}
128, & d=b \\
64, &  d=b-1 \\
32 & d=b-2 \\
16, & d=b-3 \\
8, & d < b-3.
\end{cases}
\end{align*}

The $b\times b$-matrix $L$ determines the number of buckets per level as well as how the identifier space is split among these buckets.
So the entry $L_{ij}$ of the matrix $L \in \mathbb{R}^{(b+1) \times (b+1)}$ gives the fraction of the identifiers 
at bit distance $i$ to the routing table owner for which $j$ additional digits besides the common prefix
are considered for deciding the bucket.
For example, in KAD,
$L_{b4}=1$, $L_{i3} = 0.75$, and $L_{i4} = 0.25$
for $i < b$. In the context of our probabilistic model, we consider a random variable $L_d$, whose distribution
is given by the $d$-th row of $L$.
The routing algorithm is modelled for arbitrary $\alpha$ and $\beta$ in \cite{roos13comprehending}. However, we only consider here the standard values of KAD: $\alpha=3$ and $\beta=2$.

We now describe the general idea of the derivation presented in \cite{roos13comprehending} using the above terminology.
Kademlia routing is modelled as a Markov chain, so that states correspond to the bit distance of the closest
$\alpha$ known nodes to the target identifier $x$.
Routing termination is denoted by the state $\found$. 
The state space is hence given by
\begin{align*}
S_\alpha = \{\found\} \cup \{(d_1, \ldots, d_\alpha): d_i \in \mathbb{Z}_{b+1}, d_i \leq d_{i+1}\}. \end{align*} The initial distribution $I$ gives the probability that $\alpha$ closest nodes in the requesting node's routing table have bit distances $d_1, \ldots ,d_\alpha$ to $x$.  
For any non-terminal state, the transition matrix $T$ gives the probability to get from a set of $\alpha$ contacted nodes with bit distances $d_1, \ldots ,d_\alpha$ to 
either the terminal state or a set of nodes with bit distances
 $\tilde{d}_1, \ldots ,\tilde{d}_\alpha$. 
The success rate after $k$ hops is obtained by computing 
 $T^{k-1}I$ and choosing the entry corresponding to the terminal state $\found$.
 
The model is based upon various assumptions: Node identifiers are assumed to be chosen uniformly at random,
and queries are blocking, i.e.  at each hop exactly $\alpha$ nodes are
queried before deciding on the next set of contacts to query.
The basic model also assumes a steady-state, no churn system with
maximally full buckets. 
However, it is shown in \cite{roos13comprehending} how to extend the model
to dynamic environments.
For simplicity, we stick with the basic model in the following.

The essential part of the derivation is to determine the probability distribution
of the bit distances of the closest $\gamma \in \{\alpha, \beta \}$
contacts in a node's routing table given the bit distance $D$ from the routing table owner to $x$.
The probability that the bit distances $C$ of the closest contacts corresponds to $s \in S_\gamma$ is given by
\begin{align}
\label{eq:closest}
\begin{split}
&P(C = s | D = d) \\
&= \sum_{l=0}^bP(C=s|X_0=d,L_d=l)P(L_d=l) \\
&= \sum_{l=0}^b P(C=s|X_0=d,L_d=l)L_{dl}.
\end{split}
\end{align}
The initial distribution, giving the probability of all possible distances of the requesting node, is obtained from Eq. \ref{eq:closest} by summarizing over all possible distances, i.e.
$P(X_0 = s) = \sum_{d=0}^b P(X_0 = s | D = d) P(D=d)$ with $P(D=d)=\frac{2^{d}-2^{d-1}}{2^b}$.
Similarly, the transition probabilities can be obtained from Eq. \ref{eq:closest}.
There are $\alpha\beta$ returned contacts of which the closest $\alpha$
\emph{distinct} contacts need to be selected. Let $R$ denote the
$\alpha\beta$ returned bit distances, 
$Z$ the set of all possible returned bit distances, and 
$top_{\alpha}: Z \rightarrow S_\alpha$  the shortest $\alpha$
distances of distinct contacts.  
Then the transition probability is expressed by
\begin{align}
\label{eq:transition}
&P\big(X_1  = s | 
(X_0 = (d_1, ...,d_\alpha)\big) \nonumber \\
=& \sum_{z \in Z} P(top_\alpha(z)=s|X_0 = (d_1, ...,d_\alpha), R=z)  \\
\cdot &P(R=z | X_0 = (d_1, ...,d_\alpha))  
 \prod_{j=1}^\alpha P\Big(C=z_j| D=d_j\Big) \nonumber.
\end{align}
 
The only component of the above derivation influenced by our changes is the probability $P(C = s | D = d,L_d=l)$ in Eq. \ref{eq:closest}. Our first change to the original model \cite{roos13comprehending} (allowing queries to any identifier) affects only the last step of the routing. There is no guarantee that the bucket covering the region that $x$ belongs to contains one node, whereas under the assumption of maximally full routing tables, this is given when $x$ corresponds to a node identifier. We hence consider the case $C=\found$ both for the standard scheme and for our modified scheme in Section \ref{sec:modelsuccess}.
The second change (selecting contacts such that the prefixes in the buckets are maximally diverse) modifies the probability $P(C = s | D = d)$ for all states $s$. We hence derive $P(C = s | D = d)$ for non-terminal states $s$ in Section \ref{sec:modeltransition}, distinguishing the case of $\beta=2$ and $\alpha=3$.

\subsection{Assumptions and Notations}
In this section, we state our assumptions and notations for the model. Our assumptions are mostly identical to those in \cite{roos13comprehending}, so we refer to that
publication for an in-depth discussion of their impact.
\begin{enumerate}
\item The system is in steady state without churn, failures, and attacks. In particular, there are no stale contacts in the routing tables and nodes do not fail nor do they drop messages. Furthermore, buckets are maximally full, i.e. if a bucket contains $k_1 < k$ values, there are exactly those $k_1$ nodes in the region the bucket is responsible for.
\item Node identifiers are uniformly and independently distributed over the whole identifier space. 
Requested identifiers are also chosen uniformly at random from the whole ID space. 
\item Routing table entries are chosen independently. 
\item The lookup uses strict parallelism, i.e. a node awaits all answers to its queries before 
sending additional ones.
\end{enumerate} 

\vspace{3pt}
The following notations are used throughout the derivation: 
\begin{itemize}
\item $C^{norm}$ and the $C^{own}$ denote the closest contact distributions in both the standard scheme and our modified scheme, respectively.
\item $\alpha$ is the degree of parallelism, $\beta$ the number of returned closest contacts when queried
for an identifier, $k_d$ is the bucket size at distance $d$.
\item The number of additional bits considered for the replacement scheme is given by 
$q_d=\lfloor log~k_d \rfloor$. Furthermore, we generally drop the index $d$ for $k_d$ and $q_d$ if the conditioning on $d$ is explicit.
\item The probability that a binomially distributed random variable with parameters $m$ and
$p$ takes value $z$ is denoted by $\binom(m,z,p)= {m \choose z}p^z(1-p)^{m-z}$.
\item The probability that $c$ elements chosen from a set $B$ of size $b$ are chosen from a subset
$A\subset B$ of size $a$ is abbreviated by 
\begin{align*}
\over(a,b,c)=\frac{{a \choose c}}{{b \choose c}}.
\end{align*}
\item The probability that the $\gamma$ closest
contacts to $x$ within a group of size $a$ within distance $d$ to $x$ 
have bit distances $\delta_1, \ldots , \delta_\gamma$, is denoted $\best_\gamma((\delta_1, \ldots , \delta_\gamma),d,a)$ as derived in \cite{roos13comprehending},
Eq. 10. The group 
is assumed
to be uniformly selected from all nodes within distance at most $d$. 
\item  The probability that there are no nodes in a region with $2^d$ identifiers
is abbreviated by $em(d) = \left(1-2^{d-b}\right)^{n-1}$.
\end{itemize}
Before our main derivation, we discuss an approximation made during the latter steps.

\subsection{Approximating: Empty Buckets}
\label{sec:approximation}
When querying for random IDs, there might not
be node in the bucket closest to the target identifier.
In such a case, we assume that the routing table owner knows the responsible node and the
routing terminates in the next hop.
We now show that this assumption does not considerably increase the success probability.
The responsible node either belongs to a different bucket on the same level or it is at the same
bit distance to $x$ as the routing table owner, thus belonging to a bucket on a higher level.
The responsible node is possibly not contained in the routing table if more than $k_d$ nodes
are in the region of the bucket that it belongs to.
However, the said bucket covers at most the same number of identifiers as the empty bucket that $x$ belongs to. For typical values of $k$ being at least 8, the probability 
of such an event is barely above $0.0001$, as detailed in the following.
We start by explaining the derivation in general before considering
MDHT, iMDHT, and KAD in detail.

Let $u$ be the routing table owner and $E_1$ denote the event that
that the bucket with the longest common prefix with $x$ 
in $u$'s routing table is empty. Furthermore, denote by $v$ the node
that is responsible for $x$, and by $E_2$ the event that $v$ is
not contained in $u$'s routing table.
We bound the probability by maximizing over all possible bit distances of $u$ to
$x$ 
\begin{align}
\label{eq:e12}
&P(E_1 \cap E_2) = P(E_1)P(E_2|E_1) \\
&\leq \max_{\{d,l\}} P(E_1|D=d,L_d=l)P(E_2|E_1,D=d,L_d=l). \nonumber
\end{align}
The probability that there are no nodes with the a common prefix length of $b-d+l$ with
$x$ is
\begin{align}
\label{eq:e1}
P(E_1|D=d,L_d=l)=\left( 1 - 2^{d-l-b}\right)^{n-1}.
\end{align}
An upper bound on the second factor is obtained by the
probability that there are more nodes than the bucket size with the same common
prefix as $v$.
The probability depends on the routing table structure.
For MDHT and iMDHT, which do not consider any more bits after the first non-common
bit for the routing table structure, an upper bound is given by the probability
that any bucket on higher levels in $u$'s routing table cannot contained all nodes in their respective regions.
Note that by conditioning on $E_1$, all nodes have chosen an identifier within
the remaining fraction $1-2^{d-1-b}$ of the identifier space,
so that the probability to have a common prefix $b-d+i$ with $u$ is
$\frac{2^{d-i-b}}{1-2^{d-1-b}}$.
The probability of the event $A_i$ that more than $k_{d-i}$ can be chosen at level $i < d-1$ 
is hence obtained from the complementary cumulative distribution function
of a binomially distributed random variables with parameters $n-1$ and
$\frac{2^{d-i-b}}{1-2^{d-1-b}}$. 
The probability that any of these events $A_i$ hold is then obtained by a union
bound, resulting in  
\begin{align}
\label{eq:e2MDHT}
&P(E_2|E_1,D=d,L_d=1) \leq \sum_{i=2}^d P(A_i|E_1,D=d,L_d=l) \\ 
&= \sum_{i=2}^d \left(1 - \sum_{j=0}^{k_{d-i}}\binom(n-1,\frac{2^{d-i-b}}{1-2^{d-1-b}},j)\right). \nonumber
\end{align}
The bounds for MDHT and iMDHT can be computed from Eq. \ref{eq:e12}, Eq. \ref{eq:e1}, and Eq. \ref{eq:e2MDHT}.

When considering KAD, there are 5 buckets for each non-zero common prefix length, hence we multiply the above probability
with 5. 
However, these buckets are only of interest if all buckets containing nodes with common prefix length $b-d$
are empty. 
If at most $k_d$ nodes are available for the closest non-empty bucket to $x$, the responsible node is
contained in $u$'s routing table regardless of the buckets on lower levels.
There are $m \in \{5,8\}$ buckets with the same common prefix. 
Let $j$ be the index of the empty bucket with the longest common prefix with $x$.
We sort the remaining $a \in \{4,7\}$ buckets with the same common prefix to $u$ by their distance to $x$. Let $C_{ij}$ be the number of nodes in the $i-th$
closest bucket, and $p_d= \sum_{i=2}^d \left(1 - \sum_{j=0}^{k_{d-i}}\binom(n-1,\frac{2^{d-i-b}}{1-2^{d-1-b}},j)\right)$
the quantity from Eq. \ref{eq:e2MDHT}. 

\begin{figure}
\captionsetup{font=scriptsize}
\includegraphics[width=\linewidth]{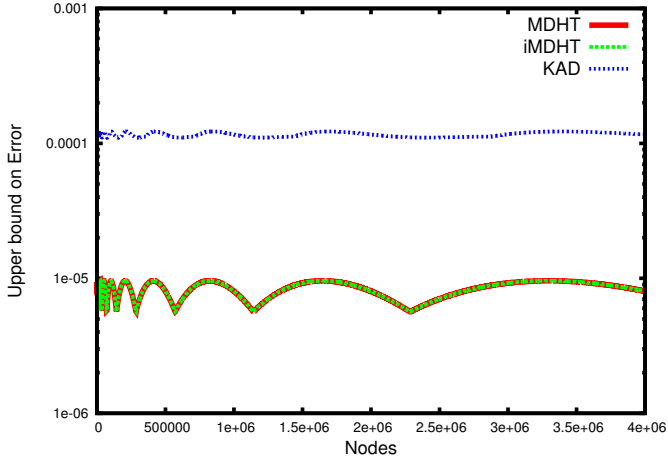}
\caption{Estimate error when assuming that the destination identifier $x$ is found if the respective routing table bucket is empty.}
\label{fig:error}
\end{figure}

Then the second probability from Eq. \ref{eq:e12} can be bound for KAD by
\begin{align}
\label{eq:e2KAD}
& P(E_2|E_1,D=d,L_d=l) \nonumber \\
&\leq  \max_{j=1,\ldots , m} P(C_{1j} > k|E_1,D=d,L_d=l) \nonumber \\
&+P(C_{1j} = 0|E_1,D=d,L_d=l) \\
&\left(P(C_{2j} > k_d|E_1,D=d,L_d=l,C_{1j}=0) \right. \nonumber \\
&+ P(C_{2j} = 0|E_1,D=d,L_d=l,C_{1j}=0) \left(\ldots + \right. \nonumber\\
&\left. \left. P(C_{aj} = 0|E_1,D=d,L_d=l,C_{1j}=0,.., C_{a-1j}=0)p_d\right)\ldots \right). \nonumber
\end{align}
The variable $C_{ij}$ given $E_1,C_{1j}=0, \ldots ,C_{i-1j}=0$ is binomially distributed with parameter $n-1$ and $q_{ij}$ chosen based on the combined fraction of identifiers that are covered by the closest $i-1$ buckets.
For $d=b$, there are 8 buckets, each responsible for the same number of identifiers, so that 
$q_{ij}=\frac{2^{-4}}{1-i2^{-4}}$. 
Otherwise, there are three distinct possibilities. If $L_d=4$ (i.e. $j=4$ or $j=5$), four bits are resolved for the closest bucket, but only
three for the remaining three buckets, so that
\begin{align*}
q_{ij} = \begin{cases} 
\frac{2^{d-4-b}}{1-2^{d-4-b}}, & i=1 \\
\frac{2^{d-3-b}}{1-2\cdot 2^{d-4-b}-(i-2)\cdot 2^{d-3-b}}, & i>1.
\end{cases}
\end{align*} 
If $l=3$, there are two possibilities. If the XOR distance of $x$ and $u$ starts with $11$ (i.e. $j=1$ or $j=2$), the closest two buckets resolve 3 bits, and $q_i$ is given by
\begin{align*}
q_{ij} = \begin{cases} 
\frac{2^{d-3-b}}{1-i\cdot 2^{d-3-b}}, & i\leq 2 \\
\frac{2^{d-4-b}}{1-3\cdot 2^{d-3-b}-(i-3)\cdot 2^{d-4-b}}, & i>2.
\end{cases}
\end{align*} 
Otherwise (i.e. $j=3$), the two closest buckets resolve 4 additional bits, and
\begin{align*}
q_{ij} = \begin{cases} 
\frac{2^{d-4-b}}{1-2^{d-3-b}-(i-1)\cdot 2^{d-4-b}}, & i\leq 2 \\
\frac{2^{d-4-b}}{1-2^{d-3-b}-2\cdot 2^{d-4-b}-(i-3)\cdot 2^{d-3-b}}, & i>2.
\end{cases}
\end{align*} 
This completes our derivation of an upper bound on the error made by assuming that a target is found if the respective bucket is empty. 

Figure \ref{fig:error} displays the upper bounds for MDHT, iMDHT and KAD, considering networks between 1,000 and 4,000,000 nodes. 
The upper bound error remains within certain bounds, reaching its maximum and minimum once for each $2^i$ additional nodes, since doubling the network size corresponds to filling one more level in the routing table.
For MDHT and iMDHT, the error probability is below $10^{-5}$. The results for the two DHTs are identical because only the higher levels have a realistic change of containing empty buckets and these have identical structures.  
The error probability for KAD is higher due to having buckets on the same level, but is barely above $10^{-4}$.

\subsection{Success Probability}
\label{sec:modelsuccess}
We first consider the case $C=\found$ both for the standard scheme and for our maximally-diverse contact selection scheme. 
The idea is to summarize over the number of possible nodes in the bucket closest to $x$.
If there are $m>0$ such nodes and at most $l\leq m$ edges to those nodes, the probability that 
one edge leads to the node responsible for $x$ is $l/m$.
If $m=0$, we assume that the responsible node is known to the routing table owner,
as motivated in Section \ref{sec:approximation}.
x.
Propositions \ref{prop:s1} and \ref{prop:s2} give the probability that the node responsible for $x$
is contained in a node's routing table both for the standard and for our modified scheme, respectively.

\vspace{3pt}
\begin{proposition}
\label{prop:s1}
The probability that a node $v$ at bit distance $d$ to the target identifier $x$ knows the responsible
node given the number of further resolved bits $L_d$ is
\begin{align}
\label{eq:Fstandard}
&P(C^{norm} = \found | D = d,L_d=l)  \approx \\
&\binom(n-1,0,2^{d-l-b}) + 
\sum_{m=1}^{n-1}\binom(n-1,m,2^{d-l-b}) \min\{1, k/m\} \nonumber
\end{align} 
when selecting contacts in a bucket uniformly at random from all nodes in the region covered by
the bucket.
\end{proposition}
\begin{proof}
The proof follows straightforward from the assumption that $n-1$ nodes (all besides the routing
table owner who is known to be at distance $d$) choose their identifier uniformly at random from
$2^{b}$ identifiers. There are $2^{d-l}$ identifiers sorted in the same bucket as $x$, so that
the probability of a node to choose such an identifier is $2^{d-l-b}$.
If there are less than $k$ identifiers in the region, the responsible node is contained in the bucket,
otherwise the contacts in the buckets are chosen uniformly from all nodes in the region, i.e.
with probability $k/m$ for $m$ nodes. \end{proof}

\vspace{3pt}
\begin{proposition} \label{prop:s2}
The probability that a node $v$ at bit distance $d$ to the target identifier $x$ knows the responsible
node given the number of further resolved bits $L_d$ is
\begin{align}
\label{eq:Fown}
&P(C^{our} = \found | D = d,L_d=l) \nonumber \\
&\approx \binom(n-1,0,2^{d-l-b}) +\sum_{m=1}^{n-1} \binom(n-1,m,2^{d-l-b}) \\
&\cdot \left(\binom(m,0,2^{-q})\min\{1, k/m\} 
+ \sum_{j=1}^{m}  \binom(m,j,2^{-q}) \rho_j \right)\nonumber 
\end{align} 
with
\begin{align*}
& \rho_j = \frac{1}{j} \\
 + 
(1-\frac{1}{j})\sum_{i=1}^{2^{q}-1} &
\frac{{2^q-1 \choose i}{m-j \choose i-1}}{{m-j+2^q-2 \choose 2^q-2}} 
& \min\left\lbrace 1,\frac{k+2^{q}-i-1}{\max\{1,m-i-1\} }\right\rbrace
\end{align*}
when maximizing the diversity of the contact prefixes in the bucket.
\end{proposition}
\begin{proof}
As in the proof of Proposition \ref{prop:s1}, we summarize over the possible number
of nodes $m$ in the same region as $x$. Furthermore, we summarize over the number $j$ of the $m$ nodes that share additional 
$q$ bits with $x$. $j$ is binomially distributed with parameters $m$ and the
$2^{-q}$, since there are $2^q$ different prefixes. 
If $j=0$, any of the $m$ nodes with a different prefix is responsible for $x$, corresponding
 to the term $\binom(m,0,2^{-q})\min\{1, k/m\}$.
If $j>0$, at least one node with the same prefix as $x$ is contained in the bucket, so the
probability that the responsible node is chosen by the one link that is guaranteed to choose
a node with such a prefix  is $1/j$.
In addition, there are links that do not have to be addressed to any specific of the $2^q$ 
prefixes if $2^q < k$ or there are prefixes that do not correspond to any existing node identifier. 
Let $i$ be the number of prefixes other than the prefix of $x$, for which there is at least one node.
Then there are $m-i-1$ nodes that are not chosen by prefix-specific links, as well as 
$k-2^q-i-1$ links that can go to any of those nodes. So if the responsible node is not chosen
because of its common prefix with $x$ (probability $1-1/j$), it can still be chosen with
probability  $\min\left\lbrace 1,\frac{k+2^{q}-i-1}{\max\{1,m-i-1\}} \right\rbrace$.
It remains to show that the probability that there are $i$ prefixes for which there is a node
is 
\begin{align*}
\frac{{2^q-1 \choose i}{m-j \choose i-1}}{{m-j+2^q-2 \choose 2^q-2}}. 
\end{align*}
Basically, we consider the problem of dividing $m-j$ nodes in $2^q-1$ equally likely regions.
An equivalent problem is to arrange $2^q-2$ borders in a set of $m-j$ elements, assuming
that all elements before the $i$-th border are assigned to region $i$, and the elements
after the border $2^q-2$ are assigned to region $2^q-1$.
So there are no elements within region $i$ if border $i+1$ follows directly after border $i$.
Using this equivalent problem formulation, the above term follows from basic
combinatoric. 
Then the total number of possibilities to arrange the $2^q-2$ borders within a set of 
$m-j+2^q-2$ objects is ${m-j+2^q-2 \choose 2^q-2}$.
The factor ${2^q-1 \choose i}$ refers to the possible arrangements of the borders
into $i$ sets without any elements between borders of the same set.
Similarly, the factor ${m-j \choose i-1}$ gives the number of
possibilities to divide $m-j$ elements into $i$ non-empty sets.\end{proof}

\subsection{Maximizing the Diversity} \label{sec:modeltransition}
The standard contact selection has been treated in \cite{roos13comprehending}, thus we here only consider our modified scheme. We first consider the case that only $\beta=2$ closest nodes are of interest, and then extend the result to $\alpha=3$ which is needed to determine the initial distribution.

Consider the case that $k=2^q$ and there are nodes for all the $2^q$ prefixes in the bucket.
Then there is exactly one contact in the bucket that has bit distance of less than
$d-l-q$ to $x$, and the node is selected uniformly at random from all these identifiers.
The second closest contact is then the one that is at bit distance $d-l-q$, i.e. the one
for which the last bit of the $q$-bit prefix is different. The third closest contact is chosen to be at bit distance $d-l-q+1$, not sharing the last two bits of the prefix with $x$. However, if $k>2^q$ or there are prefixes with no matching nodes, more than one contact can have a bit distance less than $d-l-q$. Furthermore, the next closest contacts not within distance $d-l-q$ can be farther than for the standard contact selection. 
Proposition \ref{prop:s3} and \ref{prop:s4} evaluates how likely are these scenarios, summarizing over all possible number of prefixes without matching nodes for the $\beta=2$ and $\alpha=3$ closest contacts, respectively.
\vspace{3pt}
\begin{proposition}
\label{prop:s3}
The probability that the two closest nodes to $x$ in the routing table of a node $v$
at bit distance $D=d$ to $x$ are at distances $(\delta_1,\delta_2)$ to $x$ is
\begin{align}
\label{eq:beta1}
&P\left(C^{own}=(\delta_1, \delta_2)|D=d,L_d=l, C \neq \found\right) = 0
\end{align}
if $\delta_1 \geq d-l-q$, and setting $\eta=d-l-\delta_2$
\begin{align}
\label{eq:beta2}
&P\left(C^{own}=(\delta_1, \delta_2)|D=d,L_d=l,C \neq \found\right) \nonumber \\
 \approx \sum_{r=0}^{2^q-1}&\binom(2^q-1,r,em(\max\{d-l-q-1,0\})) \nonumber \\
\cdot &\binom(r+k-2^q,0,1/(2^q-r)) \\
\cdot &\best_1((\delta_1),\max\{d-l-q-1,0\},1) \nonumber \\
&\over(r,2^{q}-1,2^{q-\eta}-1)-\over(r,2^{q}-1,2^{q-\eta+1}-1) \nonumber 
\end{align}
if $\delta_2 \geq d-l-q > \delta_1$, and
\begin{align}
\label{eq:beta3}
\begin{split}
&P\left(C^{own}=(\delta_1, \delta_2)|D=d,L_d=l,C \neq \found\right) \\
\approx \sum_{r=0}^{2^q-1}&\binom(2^q-1,r,em(\max\{d-l-q-1,0\}))\\
\cdot &\sum_{a=0}^{r+k-2^q} \binom(r+k-2^q,a,1/(2^q-r))\\
&\quad \cdot \best_2((\delta_1,\delta_2),\max\{d-l-q-1,0\},a+1)
\end{split} 
\end{align}
if $\delta_2 < d-l-q$.
\end{proposition}
\begin{proof}
Eq. \ref{eq:beta1} holds since at least one node within bit distance $d-l-q-1$ of
$x$ is chosen.
For the remaining two cases, we summarise over the number $r$ of prefixes
without any matching node. These are approximately binomially distributed with parameters $2^q-1$ for the number of other prefixes, and the probability
that there is no node with the respective common prefix (strictly speaking the
probability that common prefixes are not taken by a node are not independent of,
hence the approximation).
Given $r$, the number of additional contacts $a$ within distance at least
$d-l-q-1$ is binomially distributed with parameters $k-2^q+r$, the number
of potential additional contacts, and $1/(2^q-r)$, the probability that exactly
the prefix of $x$ is also a prefix of the identifier of the contact.
Eq. \ref{eq:beta2} considers the case $a=1$: 
If only one link leads to a node with the same prefix as $x$, the 
first at distance $\delta_1$ with probability $\best_1((\delta_1),\max\{d-l-q-1,0\},a+1)$. 
The second node is chosen as the closest of the $2^q-1$ remaining prefixes. 
Note that there are $2^{q-i}-1$ prefixes that agree with $x$'s
prefix in the first $i$ bits. 
If the closest prefix of an existing node identifier is at distance $\delta_2$, we have
$i=\eta+1$, and all closer $2^{q-\eta}-1$ prefixes are chosen from the $r$ prefixes
without any node identifier, but the closest $2^{q-\eta+1}-1$ are not chosen from
those $r$.  The probability that $m \in \{2^{q-\mu}-1, 2^{q-\mu+1}\}$ prefixes are chosen from a set of $r$
given that there are $2^q-1$ prefixes to choose from is given by
$\over(r,2^q-1,m)$, so that Eq. \ref{eq:beta2} follows.
If $a$, the number of additional contacts within distance at least $d-l-q-1$, is at least $1$, there are at least two nodes within distance $d-l-q-1$ of $x$,
and the probability distribution of their distance is given by 
$\best_2((\delta_1,\delta_2),\max\{d-l-q-1,0\},a+1)$ as derived in \cite{roos13comprehending}.
This accounts for Eq. \ref{eq:beta3} and completes the proof.\end{proof}

\vspace{3pt}
\begin{proposition}
\label{prop:s4}
Assume that the prefix diversity of the routing table entries is maximized.
The probability that the three closest nodes to $x$ in the routing table of a node $v$
at bit distance $D=d$ to $x$ are at distances $(\delta_1,\delta_2, \delta_3)$ to $x$ is
\begin{align}
\label{eq:alpha1}
&P\left(C^{own}=(\delta_1, \delta_2, \delta_3)|D=d,L_d=l\right) = 0
\end{align}
if $\delta_1 \geq d-l-q$, and setting $\eta_i=d-l-\delta_i$, 
as well as
\begin{align*}
p(r) = \frac{{2^q -1-r \choose 1}{r - 2^{q-\eta_2}+1 \choose 2^{q-\eta_2}-1}}
{{2^q-2^{q-\eta_2} \choose 2^{q-\eta_2}}},
\end{align*}
\begin{align}
\label{eq:alpha2}
&P\left(C^{own}=(\delta_1, \delta_2,\delta_3)|D=d,L_d=l\right) \nonumber\\
 =& \sum_{r=0}^{2^q-2}\binom(2^q-1,r,em(\max\{d-l-q-1,0\})) \nonumber\\
 &\quad\cdot \binom(r+k-2^q,0,1/(2^q-r)) \nonumber\\
&\quad \cdot \best_1((\delta_1),\max\{d-l-q-1,0\},1)\\
&\quad \cdot \over(r,2^q-1,2^{q-\eta_2}-1) \nonumber\\
&\begin{cases}
& \left[ 
\left(1-\over(r-2^{q-\eta_2}+1,2^q-2^{q-\eta_2}, 2^{q-\eta_2})-p(r)\right) \right. \\
&\left. +p(r)(1-\binom(k-2^q+r,0,\frac{1}{2^q-r-1})\right], \\
&\text{if }\eta_2=\eta_3 \\
&p(r) \binom(k-2^q+r,0,\frac{1}{2^q-r-1}) \\
&\left(\over(r-2^{q-\eta_2}+2,2^q-2^{q-\eta_2+1}, \sum_{i=q-\eta_2+2}^{q-\eta_3}2^i) \right. \\
& \left. -  \over(r-2^{q-\eta_2}+2,2^q-2^{q-\eta_2+1}, \sum_{i=q-\eta_2+2}^{q-\eta_3+1}2^i)\right) \\
&\text{if }\eta_2\neq \eta_3 \\
\end{cases}
\nonumber
\end{align}
if $\delta_3 \geq \delta_2 \geq d-l-q > \delta_1$, 
\begin{align}
\label{eq:alpha3}
\begin{split}
&P\left(C^{own}=(\delta_1, \delta_2, \delta_3)|D=d,L_d=l\right) \\
=&\sum_{r=0}^{2^q-2}\binom(2^q-1,r,em(\max\{d-l-q-1,0\}))\\
\cdot & \binom(r+k-2^q,1,1/(2^q-r))\\
\cdot &\best_2((\delta_1,\delta_2),\max\{d-l-q-1,0\},2) \\
\cdot &\left(\over(r,2^{q}-1,2^{q-\eta_3}-1)-\over(r,2^{q}-1,2^{q-\eta_3+1}-1)\right)
\end{split} 
\end{align}
if $\delta_2 < d-l-q\leq \delta_3$, and
\begin{align}
\label{eq:alpha4}
\begin{split}
&P\left(C^{own}=(\delta_1, \delta_2)|D=d,L_d=l\right) \\
=&\sum_{r=0}^{2^q-1}\binom(2^q-1,r,em(\max\{d-l-q-1,0\}))\\
\cdot &\sum_{a=2}^{r+k-2^q} \binom(r+k-2^q,a,1/(2^q-r))\\
&\quad \cdot \best_3((\delta_1,\delta_2,\delta_3),\max\{d-l-q-1,0\},a+1)
\end{split} 
\end{align}
if $\delta_3 < d-l-q$.
\end{proposition}
 \begin{proof}
 The only difference in contrast to returning two contacts is that potentially two nodes
 are not chosen from the set of nodes with bit distance at most $d-l-q-1$, i.e.
 Eq. \ref{eq:alpha2}.
 These can be either at the same bit distance to $x$ or not.
 In the first case, there are two possibilities: (\emph{i}) the prefixes of
 those two nodes are different. Then  at least two of the
 $2^{q-\eta_2}$ prefixes for nodes at distance $\delta_2$ are 
 prefixes to existing nodes. 
 The probability is computed by subtracting the probability that there
 are none or one such prefix from 1. (\emph{ii}) if the prefixes are the same, at least one of the $k-2^q+r$ remaining contacts 
  has the same prefix. Given that there is only one contact in the bucket with the
  same prefix as $x$, the $k-2^q+r$ contacts are chosen from $2^q-1-r$ prefixes,
  leading to the term $1-\binom(k-2^q+r,0,\frac{1}{2^q-r-1})$.
  This completes the first case in Eq. \ref{eq:alpha2}. If $\delta_2\neq \delta_3$, the probability is obtained by considering
  the events:
  \begin{itemize}
  \item $A_1$: there is exactly one prefix within the $2^{q-\eta_2+1}-1$ closest prefixes
  that is a prefix of a node, but there are no further contacts selected with that prefix,
  \item $A_2$: there are at least two nodes within the $2^{q-\eta_3}$ closest prefixes,
  \item $A_3$: there are at least two nodes within the $2^{q-\eta_3+1}$ closest prefixes.
  \end{itemize}
  The second part of the case distinction in Eq. \ref{eq:alpha2} then is given by
  $P(A_2)P(A_4\setminus A_3|A_2)$. 
The actual terms follow by the same combinatoric reasoning as for Eq. \ref{eq:beta2}, conditioned   on the event that there are node identifiers with the closest $2^{q-\eta_2}-1$ prefixes. \end{proof}
  
This completes the derivation of our model. We validate the model against simulations of three Kademlia-type systems in Sec.~\ref{subsec:results_model_sim}.


\section{Evaluation} \label{sec:results}

In this Section, we aim to assess the impact of our modified neighbour selection scheme on the lookup performance in Kademlia-type systems. In particular, we answer the following three questions: (\emph{i}) How much improvement over the default lookup performance can be gained if all system nodes implement the new scheme?, and do simulations validate our model? 
(Sec.~\ref{subsec:results_model_sim}), (\emph{ii}) What is the impact of churn on the gained  improvement? (Sec.~\ref{subsec:churn}), and (\emph{iii}) How much improvement can be gained by a partial deployment of the new scheme in a real Kademlia-type system? (Sec.~\ref{subsec:results_measure}). 

The \emph{lookup's hop count} here refers to the number of edges on the shortest path traversed during the lookup process. In Kademlia, each routing hop (i.e. step) represents a transition from a set of queried contacts to either another set of queried contacts or routing termination \cite{roos13comprehending}.

\subsection{Lookup Performance of a Full Deployment Without Churn: Model vs. Simulations}
\label{subsec:results_model_sim}

We discuss here the performance results of a full deployment for our approach (i.e. all system nodes incorporate the proposed neighbour selection scheme) as predicted by our model and validate them against simulations. We focus on three exemplary Kademlia-type systems: MDHT, iMDHT, and KAD, as they are described in Sec.~\ref{sec:kademlia}.

\vspace{3pt}
\subsubsection{\textbf{Simulation environment and setup}} Note that none of the well-known P2P simulators (e.g. PeerSim \cite{Montresor09peersim}, PeerfactSim.KOM \cite{stingl11peerfactsim}, or OverSim \cite{baumgart07oversim}) has exact implementations for the aforementioned three systems. However, the modular design of the widely-used simulator, OverSim, allows to easily add new P2P overlays. 
We hence chose to develop MDHT, iMDHT, and KAD and their respective modified versions as new P2P overlays in OverSim. We use the source code of eMule as a basis for our implementation. Please note that eMule implements a \emph{loose parallel lookup} whereas our theoretical model assumes a
\emph{strict parallel lookup} (Sec.~\ref{subsec:assumptions}: Assumption 4). However, since our current model assumes no churn, the two lookup techniques shall perform similarly \cite{stutzbach06improving}. 

We performed simulations with three system sizes: 10,000, 15,000, and 20,000 nodes. 
At the beginning of simulations, nodes are added to the system until the target size is reached. After the system has stabilized (i.e. the nodes have populated their routing table buckets with contacts), the statistics of interests, i.e. the hop count and the hop count distribution, are obtained. All simulation results are averaged over ten runs.


\vspace{3pt}
\subsubsection{\textbf{Results}} Fig.\,\ref{fig:hop_count_dist} shows the resulting CDFs of hop count distributions for the three standard systems as well as for the respective modified ones, each with 10,000 nodes, both from the model predictions and from simulations. Table \,\ref{tab:average} shows the hop count values predicted by the model, and for simulations, the respective sample average values, the 95\% confidence intervals (using the Student's t-distribution), and the median values, in addition to hop count gain achieved by the modified scheme. The hop count gain values of simulations represent: (\emph{i}) the difference between $standard\_hop\_count - CI$ and $modified\_hop\_count + CI$, and (\emph{ii}) the minimum and maximum hop count gain values. These results represent the systems without churn (as assumed by the current model).

All in all, the results show the following: (\emph{i}) The three modified systems achieve improved hop counts, i.e. they outperform the respective standards systems, which confirms the utility of the proposed scheme. (\emph{ii}) The model predictions and simulation results are very close to each other, which indicates that both the model derivations and implementation of simulations are correct. These conclusions apply for the other experimented system sizes, that we exclude here due to space constraints. As for the impact of network size, the  average hop count increased in the larger sizes, which is to be expected. However,  the size had no large impact on the improvement gained by the new scheme.  

It can be seen also that the highest performance gain is achieved by iMDHT, whereas the lowest is achieved by KAD. More precisely, in this example with 10,000 nodes, iMDHT improves the hop count by a bit higher than 7\%, whereas KAD improves only about 1.5\%, and MDHT is in between by about 4.5\%. We attribute this disparity 
to the different routing table structures in the three systems: On the one hand, KAD implements some form of diversity by default, as described in Sec.~\ref{subsec:idea}, 
which limits the impact of the additional diversity enabled by the new scheme. On the other hand, MDHT and iMDHT do not have such feature in their default designs, and therefore they are expected to benefit from the new scheme more than KAD. In addition, the larger bucket sizes at the top four routing table levels (i.e. the most used ones) in iMDHT can contain more diverse contacts, and thus achieve a higher performance gain, than MDHT.

\begin{figure*}[ht]
\centering
\captionsetup{font=scriptsize}
\subfloat[MDHT]{\label{fig:out}\includegraphics[width=0.32\linewidth]{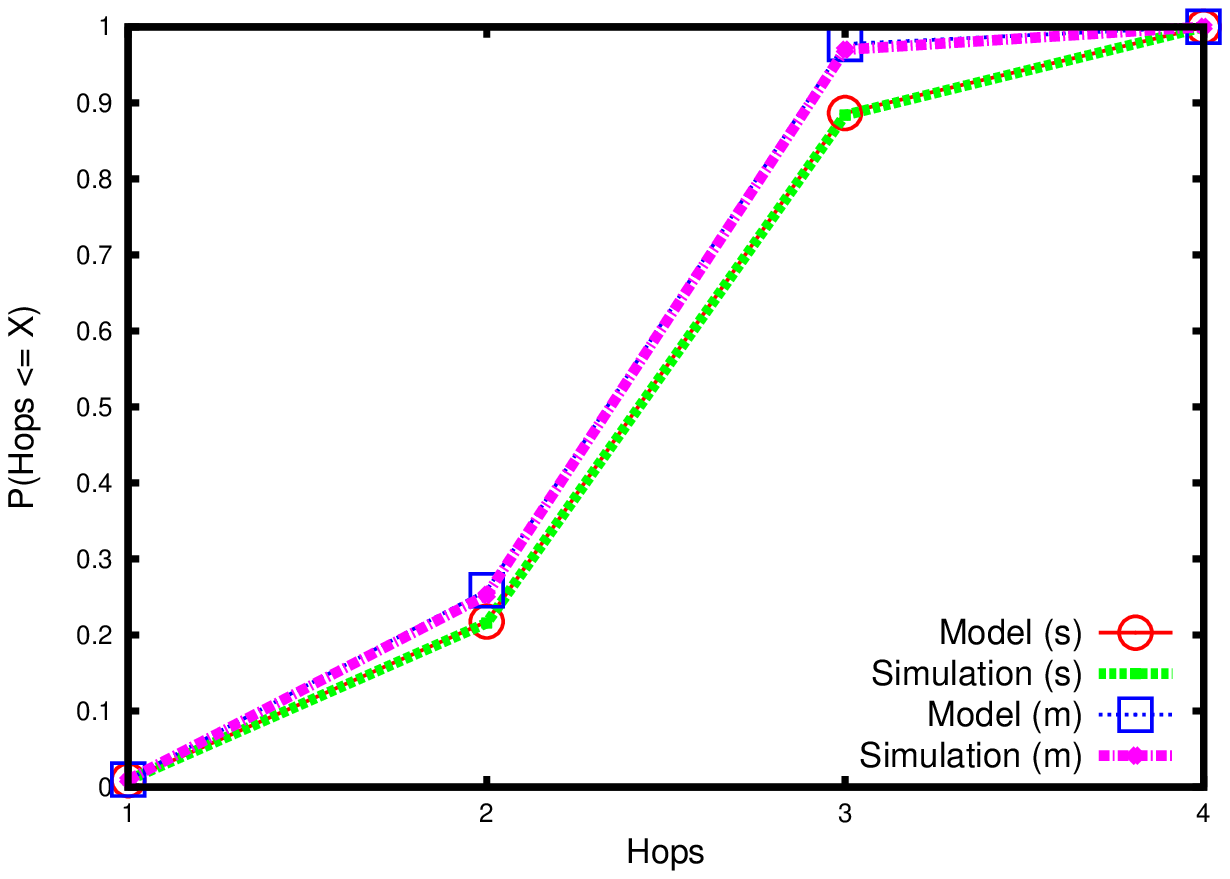}}
\subfloat[iMDHT]{\label{fig:in}\includegraphics[width=0.32\linewidth]{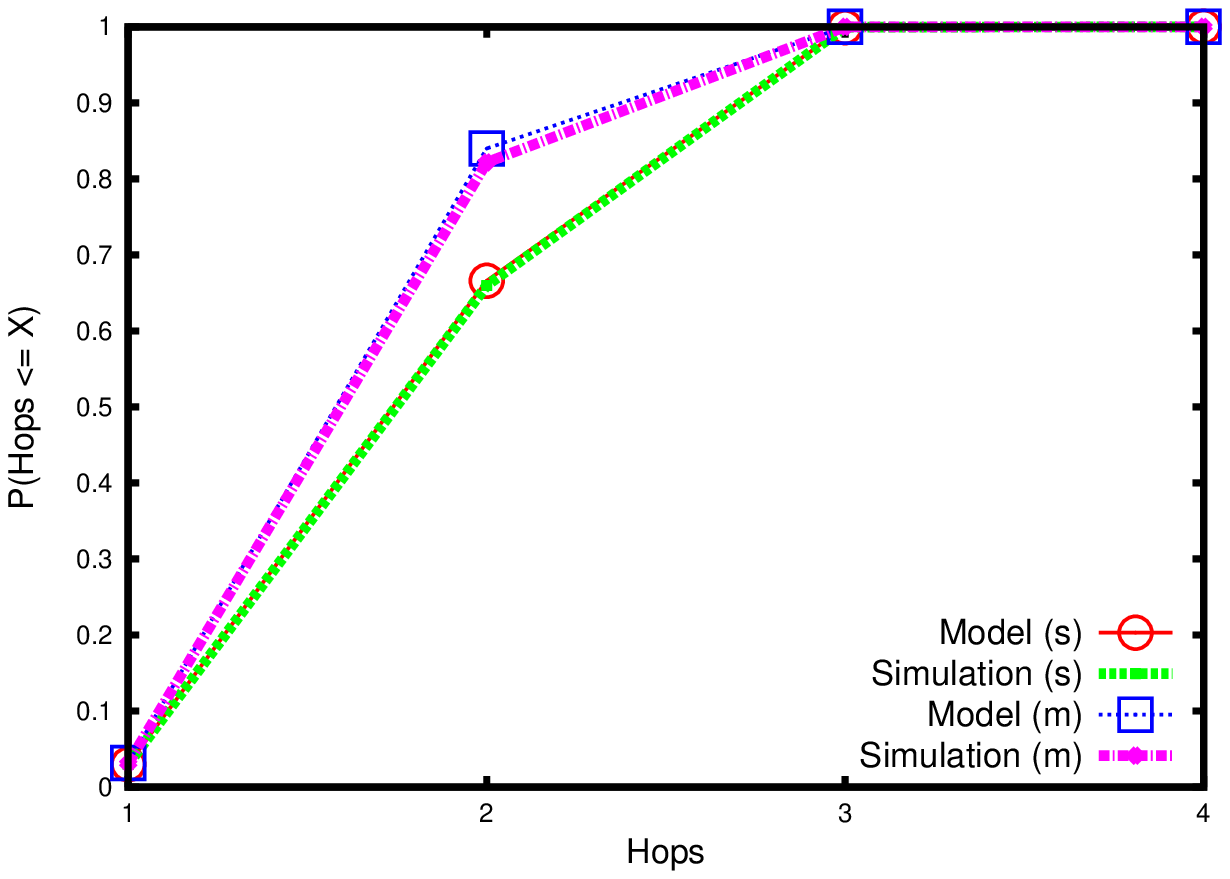}}
\subfloat[KAD]{\label{fig:dd}\includegraphics[width=0.32\linewidth]{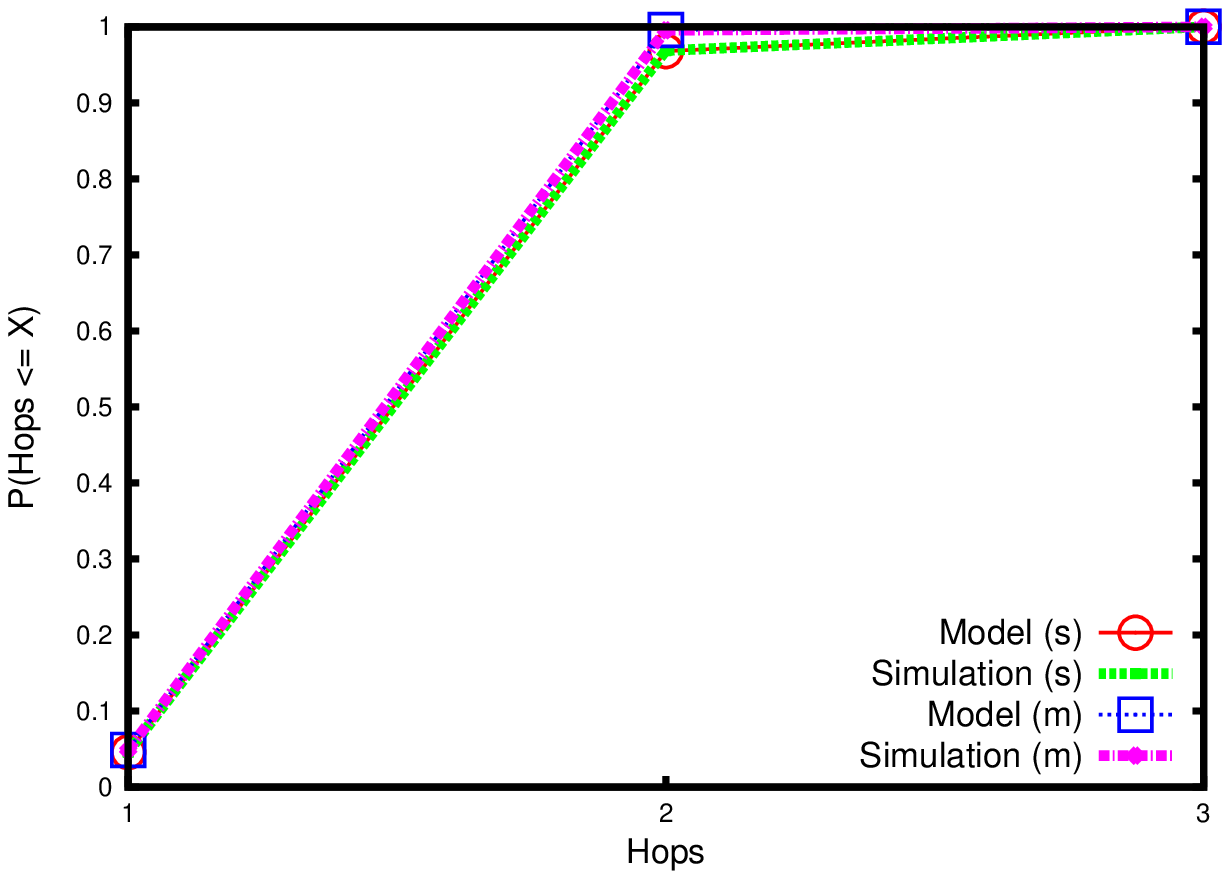}}
\caption{CDFs of hop count dist. (simulations vs. model expectations) for three exemplary Kademlia-type systems of size 10,000, without churn: (a) MDHT, (b) iMDHT, (c) KAD. \\(\textbf{s:} standard system; \textbf{m:} modified system)}
\label{fig:hop_count_dist}
\end{figure*}


\begin{table*}[t]
  \captionsetup{font=scriptsize}
 \caption{Sample average hop count with 0.95 CI (t-value) and median values for simulations vs. model expectations for three exemplary Kademlia-type systems \\ of size 10,000, without churn: Standard systems, modified systems, and the achieved hop count gain of (+ : "\textit{modified + CI}" to "\textit{standard - CI}" and [min , max]).}
\centering 
\scriptsize 

 \begin{tabular}{|l|c|c|c||c|c|c||c|c|c|}
 \cline{2-10}
 \multicolumn{1}{c|}{} & \multicolumn{3}{c||}{\textbf{MDHT}} & \multicolumn{3}{c||}{\textbf{iMDHT}} & \multicolumn{3}{c|}{\textbf{KAD}} \\
 \cline{2-10}
  \multicolumn{1}{c|}{} & \multicolumn{2}{c|}{\textbf{Simulations}} & \multirow{2}{*}{\textbf{Model}} & \multicolumn{2}{c|}{\textbf{Simulations}} & \multirow{2}{*}{\textbf{Model}} & \multicolumn{2}{c|}{\textbf{Simulations}} & \multirow{2}{*}{\textbf{Model}} \\  
  \cline{2-3} \cline{5-6} \cline{8-9}
  \multicolumn{1}{c|}{} & \textbf{Sample Avg. $\pm$ CI} & \textbf{Median} &  & \textbf{Sample Avg. $\pm$ CI} & \textbf{Median} &  & \textbf{Sample Avg. $\pm$ CI} & \textbf{Median} &  \\
 \hline

\textbf{Standard} & 2.89185 $\pm$ 0.00019 & 2.89180 & 2.88697 & 2.31113 $\pm$ 0.00032 & 2.31121 & 2.30470 & 1.98609  $\pm$ 0.00028 & 1.98607 & 1.98609 \\

 \hline
 \textbf{Modified} & 2.76774 $\pm$ 0.00074 & 2.76755 & 2.75416 & 2.14716 $\pm$ 0.00106 & 2.14699 & 2.12828 & 1.95610 $\pm$ 0.00035 & 1.95621 & 1.95535 \\
 \hline
  \multirow{2}{*}{~~\textbf{+ (\%)}} & 4.32376 & \multirow{2}{*}{-} & \multirow{2}{*}{4.60025} & 7.15366 & \multirow{2}{*}{-} & \multirow{2}{*}{7.65508} & 1.54158 & \multirow{2}{*}{-} & \multirow{2}{*}{1.54760} \\
    \cline{2-2}\cline{5-5}\cline{8-8}
  & [4.22405 , 4.33964] & & & [7.00106 , 7.19749] & & & [1.45703 , 1.58457] & & \\
 \hline
 \end{tabular}
 \label{tab:average}
 \end{table*}


\subsection{Impact of Churn} \label{subsec:churn}

We aim here to evaluate the impact of churn on the lookup performance of systems incorporating the new neighbour selection scheme. As mentioned in Sec.~\ref{sec:model}, the current version of our model supports only static scenarios (i.e. without churn). We hence perform the evaluations only by simulations. More precisely, the simulations apply the churn model proposed in \cite{yao2006modeling}, as implemented in OverSim. 
We simulated with two different average session lengths: 20,000 seconds and 10,000 seconds. Table \,\ref{tab:average_churn} summarizes the results.

Comparing the sample average with the 95\% confidence intervals (using the Student's t-distribution) and median hop count values of the two simulation settings both to each others and to the static scenario (Table \,\ref{tab:average}), the results show the following: (\emph{i}) The results without churn are always better (i.e. achieve shorter hop counts), and those with the lower churn rate (i.e. with longer average session lengths of 20,000 seconds) outperform those with higher churn rate (10,000 seconds). These results are to be expected (e.g. see \cite{stutzbach06understanding, baumgart2012fast} for explanation). (\emph{ii}) More interestingly, it can be seen that for MDHT and KAD, the hop count gain 
increases with churn (only the gain from the new scheme, not the hop count). This can be explained as follows: Since the frequency of the routing table maintenance increases with higher churn rates, there is a chance in that case to discover more contacts, and hence to increase the buckets' diversity degrees. KAD achieves the highest improvement because 
its routing table size is not a power of 2. As for iMDHT, the observation above does not hold, i.e. the churn has almost no impact on the results. The reason for this can probably be found in the extremely large bucket size on the high levels: It takes very long to fill these buckets, and hence it is hard to keep the contacts alive under churn, so that the high stale entry rate impairs further improvement.

\begin{table*}[t]
  \captionsetup{font=scriptsize}
 \caption{Median and sample average hop count with 0.95 CI (t-value) for simulations of size 10,000, with churn (2 different average session lengths): \\Standard systems, modified systems, and the achieved hop count gain (+ : "\textit{modified + CI}" to "\textit{standard - CI}" and [min , max]).}
\centering 
\scriptsize 

\begin{tabular}{|l|c|c||c|c||c|c|}
 \cline{2-7}
 \multicolumn{1}{c|}{} & \multicolumn{2}{c||}{\textbf{MDHT}} & \multicolumn{2}{c||}{\textbf{iMDHT}} & \multicolumn{2}{c|}{\textbf{KAD}} \\
 \hline
  ~~~~~~~~\textbf{Avg. session time} & \textbf{20,000 sec.} & \textbf{10,000 sec.} & \textbf{20,000 sec.} & \textbf{10,000 sec.} & \textbf{20,000 sec.} & \textbf{10,000 sec.} \\
 \hline
 ~~~~~~~\textbf{Standard (Median)} & 3.21465 & 3.32382 & 2.57720 & 2.69044 & 2.21886 & 2.31689 \\
 \hline
 ~~~~~~~\textbf{Modified (Median)} & 3.01307 & 3.11266 & 2.38905 & 2.49321 & 2.09163 & 2.17313 \\
 \hline
\textbf{Standard (Sample Avg. $\pm$ CI)} & 3.21380 $\pm$ 0.00109 & 3.32362 $\pm$ 0.00092 & 2.57716 $\pm$ 0.00029 & 2.69084 $\pm$ 0.00183 & 2.21842 $\pm$ 0.00103 & 2.31668 $\pm$ 0.00121 \\
 \hline
\textbf{Modified (Sample Avg. $\pm$ CI)} & 3.01311 $\pm$ 0.00025 & 3.11264 $\pm$ 0.00203 & 2.38908 $\pm$ 0.00161 & 2.49309  $\pm$ 0.00128 & 2.09151 $\pm$ 0.00580 & 2.17308 $\pm$ 0.00182 \\
 \hline
  \multirow{2}{*}{~~~~~~~~~~~~~~\textbf{+ (\%)}} & 6.20486 & 6.26100 & 7.22454 & 7.23849 & 5.41489 & 6.07047 \\
  \cline{2-7}
  & [6.16229 , 6.30019] & [6.19498 , 6.49519] & [7.14763 , 7.44761] & [7.22199 , 7.49043] & [4.77028 , 6.19054] & [6.00419 , 6.34748] \\
 \hline
 \end{tabular}
 \label{tab:average_churn}
 \end{table*}
 

\subsection{Lookup Performance of a Partial Deployment: Measurements of Modified KAD Clients} \label{subsec:results_measure}

We provide here an additional evaluation for the impact of the new scheme on the lookup performance by measuring it in a real Kademlia-type system. However, a large-scale deployment of the new scheme in a lifelike Kademlia-type system does not exist. Therefore, measuring the performance gain of a full deployment is infeasible. Instead, we measured the lookup performance on modified nodes (implementing the new scheme), during their participation in a standard Kademlia-type system. In this setting, the modified nodes are expected to benefit from the modified scheme, but the 
benefit is less pronounced as for a complete implementation.



\vspace{3pt}
\subsubsection{\textbf{Measurement environment and setup}} We performed our measurements in KAD, using two clients: the first uses the standard KAD code as implemented in the eMule software
, while the second implements the new neighbour selection scheme. 

During each measurement run, each client issued 500 lookup requests, one request every three seconds. The target identifiers are selected such that they are uniformly distributed over the identifier space. In particular, we used 500 keywords from the list of Steiner et al. \cite{steiner10evaluating}. 
During measurements, the clients recorded the following information: (\emph{i}) the number of hops traversed by each lookup request, and (\emph{ii}) the diversity degrees of routing table buckets at level 4. 
As described in Sec.~\ref{subsec:real_diverse}, those buckets, on average, are the mostly used ones.

\vspace{3pt}
\subsubsection{\textbf{Results}} 

We performed 40 measurement runs at several times of the day. 
Fig.\,\ref{fig:all_hop_count_measure_dist} shows the CDFs of hop count distributions both for the standard client and for the modified client. 
The median values for the standard client and the modified client were: $3.39669$ hops and $3.22624$ hops, respectively. The corresponding sample average with the 95\% confidence intervals (using the Student's t-distribution) were: $3.39862$ hops ($\pm 0.0.01583$) and $3.23172$ hops ($\pm 0.01310$). In this example, the achieved hop count reduction (computed as the difference between $standard\_hop\_count - CI$ and $modified\_hop\_count + CI$) is $4.07861\%$, with a minimum of $3.36571\%$ and a maximum of $7.68003\%$. 

Unfortunately, the measurement results cannot be compared to model predictions because the model does not support churn, nor to simulations because the simulator cannot scale to relatively very large sizes like the size of the real KAD system\footnote{Recently, \cite{salah13capturing} counted more than 300,000 concurrent online nodes in KAD.}. Since only our modified client implements the new neighbour selection scheme, while the rest are standard clients, one can expect the gain of the new scheme to be lower than what is achieved. However, when looking on the diversity degrees, the modified client could improve the diversity degrees of its buckets\footnote{For the reasons that we mentioned in Sec.~\ref{subsec:real_diverse}, we restricted our analysis here also only to the buckets located at the fourth routing table level.}: The modified client achieved a diversity degree of $7.37$ (which is very close from the maximal value), on average, compared to $4.68$ for the standard client (see also the results of measured diversity degrees of a large sample of standard clients in Sec.~\ref{subsec:real_diverse}: Fig.\,\ref{fig:measured_unif_cdf}). This increase in the diversity degrees can explain the achieved improvement, and confirms the utility of increasing the diversity. 

We attribute the obtained increase in the diversity degrees to the expected high frequency of performing the routing table maintenance processes in the real system, due to the high ratio of stale routing table entries. For instance, in \cite{salah14characterizing} we reported 88\% as the ratio of stale routing table entries, computed on almost complete and instantaneous graph snapshots of KAD. As explained in Sec.~\ref{subsec:churn}, the higher the churn rate in KAD (i.e. more frequent maintenance for routing tables) the higher the chance to discover new neighbours that increase the buckets' diversity degrees. 



\begin{figure} \centering
\captionsetup{font=scriptsize}
      \includegraphics[width=0.95\linewidth, height=0.18\textheight]{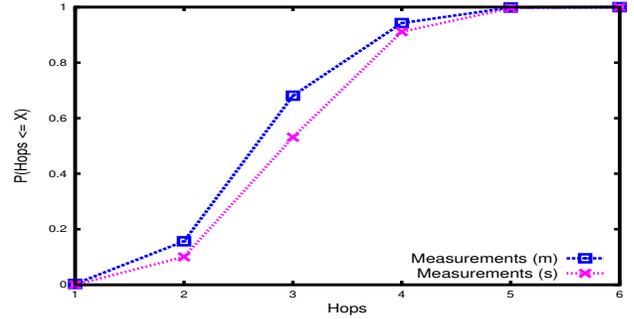} 
   \caption{CDFs of hop count distributions from measurements of standard and modified KAD clients. (\textbf{s:} standard client; \textbf{m:} modified client)}
   \label{fig:all_hop_count_measure_dist}
\end{figure}

\section{Conclusion} \label{sec:conclusion}
We proposed, modelled, and evaluated a new neighbour selection scheme for Kademlia-type systems. It aims to improve the lookup performance, almost without extra cost, by only attempting to maximize the identifiers' diversity within each routing table bucket.  


Our model predictions, in very close agreement with simulations, as well as measurements of modified KAD clients, confirm the positive impact of our scheme on the lookup performance, in form of reduction in the average hop count. The simulation results also show that the systems with small bucket sizes (namely: MDHT and KAD) can benefit more from our approach in the dynamic scenarios (i.e. with churn). We have attributed this to the resulting high frequency of routing table maintenance processes which are utilized by our scheme to improve the performance. Nevertheless, the dynamic scenario cannot be captured by the current version of the model, and therefore those results are still preliminary.

Our plans for future work include: (\emph{i}) modelling the impact of churn, network size, and the partial deployment scenario (where only part of the nodes implements our scheme), and (\emph{ii})  extending our approach by integrating it to notable prior improvements like \cite{kaune08embracing} and \cite{li2005bandwidth}.


\section*{Acknowledgements}
This work was supported by the German Academic Exchange Service (DAAD: project no. A/09/97920) and the German Research Foundation (DFG: project no. DFG GRK 1362 (TUD GKmM)). The authors would like to thank Sven Frese for his help with the measurements.


\bibliographystyle{IEEEtran}
\bibliography{main}
\end{document}